\pdfoutput=1
\documentclass{fundam}

\usepackage{enumerate}
\usepackage{amsmath}
\usepackage{amssymb}
\usepackage{url} 
\usepackage[ruled,lined]{algorithm2e}
\usepackage{graphicx}
\usepackage{tikz}
\usetikzlibrary{automata, positioning, arrows}

\newcommand{\A}{\mathcal{A}}
\newcommand{\B}{\mathcal{B}}
\newcommand{\D}{\mathcal{D}}
\newcommand{\T}{\mathcal{T}}
\newcommand{\F}{\mathcal{F}}
\newcommand{\Hy}{\mathcal{H}}
\newcommand{\dcpriorityopt}{\textsc{Don't Care Priority Optimization}}
\newcommand{\gendcpriorityopt}{\textsc{Quotient Don't Care Priority Optimization}}
\newcommand{\nat}{\mathbb{N}}
\newcommand{\infs}[1][\A]{\mathit{Inf}_{#1}}
\newcommand{\ind}{\textit{Ind}}

\newcommand{\mscc}{\text{MSCC}}

\newcommand{\comp}[1]{\overline{#1}}
\newcommand{\TLD}{\T_{L,D}}
\newcommand{\TUD}{\T_{U,D}}
\newcommand{\simLD}{\sim_{L,D}}
\newcommand{\TSf}{\T_{S,f}}

\newcommand{\IB}{\mathbb{IB}}
\newcommand{\IC}{\mathbb{IC}}
\newcommand{\IP}{\mathbb{IP}}
\newcommand{\IM}{\mathbb{IM}}

\newcommand{\MQ}{\textit{MEMBER}}
\newcommand{\EQ}{\textit{EQUIV}}

\begin{document}


\setcounter{page}{69}
\publyear{22}
\papernumber{2152}
\volume{189}
\issue{1}

\finalVersionForARXIV


\title{On Minimization and Learning of Deterministic  $\omega$-Automata \\ in the Presence of Don't Care Words}

\author{Christof L\"oding\thanks{Supported by DFG grant LO 1174/7-1}\thanks{Address for
                               correspondence:  Department of Computer Science,  RWTH Aachen University,
                               Germany. \newline \newline
                          \vspace*{-6mm}{\scriptsize{Received November 2022; \ accepted  April  2023.}}},\,   Max Philip Stachon
                    \\
  Department of Computer Science\\
  RWTH Aachen University, Germany \\
 loeding@automata.rwth-aachen.de\\
max.stachon@rwth-aachen.de}

\maketitle

\runninghead{C.~L\"oding and M.P.~Stachon}{Minimization and Learning of $\omega$-Automata with Don't Cares}

\begin{abstract}
  We study minimization problems for deterministic $\omega$-automata in the presence of don't care words. We prove that the number of priorities in deterministic parity automata can be efficiently minimized under an arbitrary set of don't care words. We derive that from a more general result from which one also obtains an efficient minimization algorithm for deterministic parity automata with informative right-congruence (without don't care words).
  We then analyze languages of don't care words with a trivial right-congruence. For such sets of don't care words it is known that weak deterministic B\"uchi automata (WDBA) have a unique minimal automaton that can be efficiently computed from a given WDBA (Eisinger, Klaedtke 2006). We give a congruence-based characterization of the corresponding minimal WDBA, and show that the don't care minimization results for WDBA do not extend to deterministic $\omega$-automata with informative right-congruence: for this class there is no unique minimal automaton for a given don't care set with trivial right congruence, and the minimization problem is NP-hard. Finally, we extend an active learning algorithm for WDBA (Maler, Pnueli 1995) to the setting with an additional set of don't care words with trivial right-congruence.
\end{abstract}


\section{Introduction}

In this paper we consider minimization and learning problems with don't care sets for deterministic automata on infinite words, called $\omega$-automata in the following. For finite words, it is well known that there is a unique minimal deterministic finite automaton (DFA) for each regular language $L$, which can be computed efficiently by merging language equivalent states in a given DFA for $L$ (see, e.g., \cite{HopcroftU79}). The minimal DFA can be characterized by the Myhill-Nerode congruence on finite words, in which two words are equivalent if they have the same extensions into the language $L$. The states of the minimal DFA correspond to the equivalence classes of this congruence. This characterization is also the basis for active learning algorithms that identify a DFA for an unknown language $L$ by equivalence and membership queries \cite{ANGLUIN198787}.

In some applications, some input words might not play a role because they do not occur and thus it does not matter whether the automaton accepts them. Such a set of input words is called a don't care set in this paper. The minimization problem with a given don't care set is then to compute a minimal automaton that behaves the same as the given one on all words that are not in the don't care set. Minimal automata in this setting need not to be unique, and the problem (its decision variant) becomes NP-hard \cite{Pfleeger73}.

For deterministic $\omega$-automata, the minimization problem (without don't cares) is harder than for DFAs. There is a variety of acceptance conditions for $\omega$-automata, the most basic one being the B\"uchi condition. A deterministic B\"uchi automaton (DBA) has a set of accepting states, like a DFA, and a run (an infinite state sequence) is accepting if an accepting state is visited infinitely often. It turns out that the minimization problem for DBA is NP-hard \cite{schewe2010beyond}, and minimal DBA for an $\omega$-language are, in general, not unique. Since a deterministic parity automaton (DPA) for a DBA-recognizable language can be turned into an equivalent DBA without changing the transition structure (as shown in \cite{KrishnanPB94} for Rabin conditions, which include parity conditions), minimization for DPA is also NP-hard. In a parity condition, each state is assigned a priority (a natural number), and a run is accepting if the maximal priority that occurs infinitely often is even. DPAs are interesting for algorithmic use because they capture the full class of regular $\omega$-languages (see \cite{THOMAS1990133,2001automata} for background on $\omega$-automata), and their acceptance condition is specified in a very compact way. For this reason we mainly study DPA and subclasses thereof in this paper (note that a B\"uchi condition can be expressed with priorities $2$ for accepting and $1$ for non-accepting).

The only class of deterministic $\omega$-automata for which an efficient minimization algorithm is known is the one of weak deterministic B\"uchi automata (WDBA) \cite{LODING2001105}, in which every strongly connected set consists only of accepting states, or only of non-accepting states. Minimal WDBA can also be characterized by the Myhill-Nerode congruence (its natural extension to infinite words) \cite{STAIGERL1974,MALER199793,WAGNER1979123}, and (a variant of) the active learning algorithm for DFAs from \cite{ANGLUIN198787} has been adapted to WDBA \cite{MALER1995316}. As for minimization, this is the only class of $\omega$-automata for which a polynomial time active learning algorithm is known. Other active learning algorithms for regular $\omega$-languages as in \cite{AngluinF16} learn different representations that can only be translated to deterministic $\omega$-automata with an exponential cost.

In \cite{EisingerK06} it is shown that there is an interesting class of don't care sets for which the minimization problem for WDBAs can be solved efficiently: If the don't care set is a regular $\omega$-language with a trivial Myhill-Nerode congruence consisting only of one class (in the following, we refer to the Myhill-Nerode congruence simply as the canonical right-congruence of the language). For finite words the only languages with this property are the trivial languages (all words or the empty language). For $\omega$-words, however, all languages in which membership only depends on the infinite suffix of the word have a trivial right-congruence, for example, the set of all infinite words with infinitely many $b$. In \cite{EisingerK06} such don't care sets are used for reducing the size of WDBAs for sets of real numbers by making certain representations of real numbers don't cares.

This raises one of the questions that we study in this paper, which we answer positively: Can the active learning algorithm from \cite{MALER1995316} be extended to this setting of a don't care set with trivial right-congruence?

We also investigate whether the minimization result for WDBAs with don't care sets can be extended to larger classes of $\omega$-automata. Since the correspondence between the states of the minimal automaton and the classes of the right-congruence plays an important role in all these results, we consider the class of deterministic $\omega$-automata with informative right-congruence (abbreviated IRC) \cite{Angluin_2018}. These are automata that use only one state per equivalence class, so their transition system is isomorphic to the transition system induced by the canonical right-congruence of the accepted language. Minimal WDBA have informative right-congruence, and we study here the question whether efficient minimization of DPA with IRC is possible under a don't care set with trivial right-congruence.

Finally, we also consider the problem of optimizing the acceptance condition of a DPA under a don't care set. It is known that for a given DPA, one can efficiently compute an equivalent parity condition on the same transition structure that uses the minimal number of different priorities \cite{CartonM99}. In \cite{LodingP19b} such a priority minimization has turned out to be a useful step before applying heuristics for reducing the state space of DPA resulting from a determinization procedure. The number of priorities might be further  reducible if we are given a set of don't cares.

\medskip
Our contributions on these questions can be summarized as follows:
\begin{itemize}
\item We show that polynomial time priority minimization on the transition structure of a given DPA is possible for a don't care set that is given by a DPA (in this case, it does not need to have trivial right-congruence).
\item We show that the minimization result for WDBA under don't care sets with trivial right-congruence does not extend to DPA with IRC: Given a DPA with IRC and a don't care set with trivial right-congruence, there is, in general, not a unique minimal DPA under this don't care set, and (the decision variant of) this don't care minimization problem is NP-hard.
\item The active learning algorithm for WDBAs can be adapted to the setting with a don't care set with trivial right-congruence.
\end{itemize}

The remainder of this paper is structured as follows.
In Section~\ref{sec:prelim} we introduce basic notation and terminology. In Section~\ref{sec:priority-optimization} we consider the problem of optimizing parity conditions under a given set of don't care words. The state minimization problem for don't care sets with trivial right-congruence is discussed in Section~\ref{sec:state-minimization}, and in Section~\ref{sec:learning} we present the active learning algorithm for WDBAs under a don't care set with trivial right-congruence. In Section~\ref{sec:conclusion} we conclude.

\section{Preliminaries}\label{sec:prelim}

An alphabet $\Sigma$ is a finite set whose elements are called letters or symbols. By $\Sigma^*$ we denote the set of all finite words over $\Sigma$, by $\epsilon$ the empty word, $\Sigma^+:=\Sigma^*\backslash \{\epsilon\}$, and $\Sigma^\omega$ denotes the set of all $\omega$-words over $\Sigma$. For $\alpha \in \Sigma^\omega$ we write $\alpha(i)$ for the letter of $\alpha$ at position $i$. An $\omega$-word $\alpha \in \Sigma^\omega$ is ultimately periodic if $\alpha=uv^\omega:=uvvv \dots $ for some $u,v \in \Sigma^*$ with $v \neq \epsilon$.
An $\omega$-language is a subset of $\Sigma^\omega$, simply referred to as language in the following. The complement of a language is denoted by $\overline{L} = \Sigma^\omega \setminus L$.

\medskip
Deterministic $\omega$-automata consist of a deterministic transition graph (or deterministic transition system), and an acceptance condition. The deterministic transition graph is a tuple $\T=(\Sigma,Q,\delta,q_0)$, where $\Sigma$ is an alphabet, $Q$ is a finite set of states, $\delta:Q \times \Sigma \rightarrow Q$ is a transition function, and $q_0$ is the initial state. We define the extended transition function as $\delta^*: \Sigma^* \rightarrow Q$ with $\delta^*(q,\epsilon):=q$ and $\delta^*(q,\sigma w)=\delta^*(\delta(q,\sigma),w)$ for $q \in Q$, $\sigma \in \Sigma$, and $w \in \Sigma^*$. We sometimes write $\delta^*(w)$ for $\delta^*(q_0,w)$. Further, we assume that all states are reachable from the initial state.

We say that two states $p,q \in Q$ are strongly connected if we can reach $p$ from $q$ and vice versa, that is, there are $u,v$ with $\delta^*(p,u) = q$ and $\delta^*(q,v) = p$. A strongly connected component (SCC) is a subset of states that are strongly connected, and a maximal strongly connected component (MSCC) is an SCC that is not a proper subset of any other SCC. We define $\mscc(q):=\{p \in Q\mid p \text{ is in the same MSCC as } q\}$.

For $\alpha \in \Sigma^\omega$, the run of $\alpha$ in $\T$ is an infinite sequence of states $\rho \in Q^\omega$ with $\rho(0)=q_0$ and $\rho(i+1)=\delta(\rho(i),\alpha(i))$ for all $i \in \nat$. The set of states that occur infinitely often in this run is called the infinity set, and is denoted by $\mathit{Inf}(\rho):=\{q \in Q: \rho(i)=q \text{ for infinitely many } i \in \nat\}$. Instead of referring to $\rho$ directly, we often write $\infs[\T](\alpha)$ for the set of states occurring infinitely often in the run of $\alpha$ in $\T$ (we also use this notation for automata instead of transition systems).

Below we define different standard types of $\omega$-automata.
A deterministic B\"uchi automaton (DBA) $\mathcal{A}=(\Sigma,Q,\delta,q_0,F)$ has an acceptance condition in the form of a set $F \subseteq Q$ of accepting (or final) states. An $\omega$-word is accepted by the DBA if and only if it produces a run that visits an accepting state infinitely often.
More formally, we say that $\mathcal{A}$ recognizes the $\omega$-language $L(\mathcal{A})=\{\alpha \in \Sigma^\omega\mid \infs(\alpha) \cap F \neq \emptyset\}$.
The class of DBA-recognizable $\omega$-languages is denoted as $\mathbb{DB}$.

A deterministic co-B\"uchi automaton (DCA) $\mathcal{B}=(\Sigma,Q,\delta,q_0,F)$ also has a set $F \subseteq Q$ of accepting states and recognizes the $\omega$-language $L(\mathcal{B})=\{\alpha \in \Sigma^\omega\mid \infs(\alpha) \subseteq F\}$.
In other words, a DCA accepts an $\omega$-word if and only if it produces a run that exclusively visits accepting states infinitely often.
We denote the class of DCA-recognizable $\omega$-language as $\mathbb{DC}$.
Note that for a DBA $\mathcal{A}=(\Sigma,Q,\delta,q_0,F)$, the DCA $\mathcal{B}=(\Sigma,Q,\delta,q_0,Q\backslash F)$ recognizes $\overline{L(\A)}$ (while the classes $\mathbb{DB}$ and $\mathbb{DC}$ are not closed under complement, see \cite{THOMAS1990133}).

A weak deterministic B\"uchi automaton (WDBA) is a DBA with the additional property that each MSCC is either contained in $F$ or disjoint from $F$.
It follows that any WDBA can be interpreted as a DCA recognizing the same $\omega$-language. Furthermore, the intersection $\mathbb{DB} \cap \mathbb{DC}$ defines precisely the class of WDBA-recognizable $\omega$-languages \cite{STAIGERL1974,MALER199793,WAGNER1979123}.

A deterministic parity automaton (DPA) $\mathcal{A}=(\Sigma,Q,\delta,q_0,c)$ has an acceptance condition in the form of a priority function $c:Q\rightarrow \nat$.
An $\omega$-word is accepted by the DPA if and only if it produces a run such that the highest priority of any state that is visited infinitely often is even.
More formally, we say that $\mathcal{A}$ recognizes the $\omega$-language $L(\mathcal{A})=\{\alpha \in \Sigma^\omega\mid \max\{c(q)\mid q \in \infs(\alpha)\}\text{ is even.}\}$.
For a set $P\subseteq Q$ of states we let $c(P) = \{c(p) \mid p \in P\}$.
For a deterministic Muller automaton (DMA) $\mathcal{A}=(\Sigma,Q,\delta,q_0,\mathcal{F})$ the acceptance condition is given as the set of accepting infinity sets $\mathcal{F} \subseteq 2^Q$ with
$L(\mathcal{A})=\{\alpha \in \Sigma^\omega\mid \infs(\alpha) \in \mathcal{F}\}$.
We refer to the class of DPA-recognizable $\omega$-languages as $\mathbb{DP}$ and the class of DMA-recognizable $\omega$-languages as $\mathbb{DM}$.

Note that Muller conditions are the most general ones since all the other conditions that are defined here can be expressed as Muller condition.
For algorithmic questions, however, Muller conditions are not well suited because a Muller condition needs to explicitly list all accepting infinity sets. For this reason, we mainly consider WDBA, DBA, DCA, and DPA in this article. It is well known that DPA are as expressive as DMA, and that DBA and DCA are less expressive, that is,  $\mathbb{DB} \cup \mathbb{DC} \subsetneq \mathbb{DP}=\mathbb{DM}$ (see, e.g., \cite{2001automata}). An $\omega$-language is regular if it can be accepted by a DPA or DMA.

A deterministic $\omega$-automaton is called minimal if there is no deterministic $\omega$-automaton of the same type with less states that accepts the same language. Besides the number of states, there is also the question of the number of priorities used in a DPA. We thus call the priority function of a DPA optimal if there is no priority function on the same transition graph that uses less priorities and accepts the same language. Note that one could also use the number of transitions as criterion for minimality, but for (complete) deterministic automata there is one transition for every combination of state and input letter, so these two measures for minimality are equivalent in our setting.

A right-congruence (on finite words) is an equivalence relation ${\sim} \subseteq \Sigma^* \times \Sigma^*$ with the property that $u \sim v$ implies $u\sigma \sim v\sigma$ for all $u,v \in \Sigma^*$ and $\sigma \in \Sigma$. We write $[u]_\sim$ for the equivalence class of $u$ in $\sim$. If $\sim$ is clear from the context, we simply write $[u]$. The index of $\sim$, which is the number of equivalence classes of $\sim$, is denoted by $\ind(\sim)$.
A right-congruence $\sim$ of finite index induces a finite deterministic transition system $\T_\sim = (\Sigma, Q_\sim, \delta_\sim,[\epsilon]_\sim)$ with $Q_\sim = \{[u]_\sim \mid u \in \Sigma^*\}$ and $\delta_\sim([u],\sigma) = [u\sigma]_\sim$. Vice versa, each deterministic transition system induces a right congruence in which two words are equivalent if they lead to the same state.

For $L \subseteq \Sigma^\omega$, the canonical right-congruence $\sim_L$ of $L$ is defined by $u \sim_L v$ if and only if \linebreak
$(u\alpha \in L \Leftrightarrow v\alpha \in L)$
for all $\alpha \in \Sigma^\omega$.
A deterministic $\omega$-automaton that accepts a language $L$ has informative right-congruence (IRC) if its transition graph is isomorphic to $\T_{\sim_L}$.
It is not hard to see that, in general, regular $\omega$-languages cannot be accepted by automata with informative right-congruence  (see \cite{Angluin_2018} for a detailed analysis of automata with IRC).
By $\mathbb{IB}$, $\mathbb{IC}$, $\mathbb{IP}$, and $\mathbb{IM}$ we denote the classes of $\omega$-languages that can be accepted by a deterministic $\omega$-automaton with IRC and the corresponding accepting condition (B\"uchi, co-B\"uchi, parity, Muller, respectively).
Note that in any DMA (and hence also DBA, DCA, DPA) for $L$, two words $u,v \in \Sigma^*$ leading to the same state must be $\sim_L$-equivalent. Thus, there is a surjective homomorphism from the transition graph $\T$ of the DMA into $\T_{\sim_L}$. The notion of homomorphism for two transition graphs $\T_1,\T_2$ with $\T_i = (\Sigma,Q_i,\delta_i,q_{0,i})$ is defined as usual. It is a function $h:Q_1 \rightarrow Q_2$ with $h(q_{0,1}) = q_{0,2}$ such that $h(\delta_1(q,a)) = \delta_2(h(q),a)$ for all $q \in Q_1$ and $a \in \Sigma$. If there is a surjective homomorphism from $\T_1$ to $\T_2$, we also say that $\T_2$ is a quotient of~$\T_1$.

\emph{Don't care words} are redundant $\omega$-words that we need not consider, that is to say we do not care whether these words are accepted or rejected by an $\omega$-automaton. Let $L,L' \subseteq \Sigma^\omega$ be two $\omega$-languages and $D \subseteq \Sigma^\omega$ a set of don't care words. We say that $L$ and $L'$ are $D$-equivalent, denoted as $L \equiv_D L'$, if and only if $L \backslash D = L' \backslash D$. We say that a deterministic $\omega$-automaton $\mathcal{A}$ is $D$-minimal if and only if there is no deterministic $\omega$-automaton $\mathcal{B}$ of the same type with $L(\mathcal{A})\equiv_D L(\mathcal{B})$ with less states than~$\mathcal{A}$.

\section{Priority optimization}\label{sec:priority-optimization}

In this section we consider the problem of optimizing a parity condition under a given don't care set, referred to as {\dcpriorityopt}: Given DPAs $\A,\D$ with  $\A = (\T,c)$ and $D := L(\D)$, find a parity condition $c'$ on $\T$ with the smallest number of priorities such that $L(\T,c')  \equiv_D L(\T,c)$. Figure~\ref{fig:prio-opt} shows an example of a DPA $\A$ in which one priority less can be used if the don't care set is $\Sigma^*a^\omega$.

\begin{figure}[t]
  \begin{minipage}{.56\textwidth}
\begin{tikzpicture}[node distance=3cm,thick,scale=0.7, every node/.style={scale=0.7}]
  \node(q0)[state with output, initial,initial text={}] at (0,0)
    {$q_0$ \nodepart{lower} $1$};
  \node(q1)[state with output, right of=q0]
    {$q_1$ \nodepart{lower} $2$};
  \node(q2)[state with output, right of=q1]
    {$q_2$ \nodepart{lower} $3$};
  \node(q3)[state with output, right of=q2]
    {$q_2$ \nodepart{lower} $4$};

    \path[-stealth]
    (q0) edge [loop above] node{$a$} ()
    (q0) edge  node[above]{$b$} (q1)
    (q1) edge [bend left] node[above]{$b$} (q0)
    (q1) edge  node[above]{$a$} (q2)
    (q2) edge [bend left] node[above]{$a$} (q1)
    (q2) edge  node[above]{$b$} (q3)
    (q3) edge [bend right] node[above]{$a$} (q1)
    ;
\end{tikzpicture}
  \end{minipage}
  \begin{minipage}{.44\textwidth}\small
      With $\A = (\T,c)$ on the left, $D = \Sigma^*a^\omega$, and $\T' =\T$, we obtain
    \[
    \begin{array}{l}
  \F_0' = \{\{q_0,q_1,q_2,q_3\},\{q_1,q_2,q_3\},\{q_0,q_1\}\} \\
  \F_1' = \{\{q_0,q_1,q_2\}\}
\end{array}
\]
    \end{minipage}
\caption{Illustration of {\dcpriorityopt}. An optimal parity condition consistent with $(\F_0',\F_1')$ can redefine the priority of $q_0$ to $2$ and thus use one priority less.} \label{fig:prio-opt}
\end{figure}

\medskip
We show that this problem can be solved in polynomial time. In fact, we consider a more general problem, in which the transition system on which we compute the parity condition is a quotient of the given DPA $\A$. Formally, we define the problem of generalized parity optimization under don't cares, {\gendcpriorityopt}, as follows: Given  DPAs $\A,\D$ with  $\A = (\T,c)$ and $D := L(\D)$, a deterministic transition system $\T'$, and a surjective homomorphism $h:\T \rightarrow \T'$, decide if there is a parity condition $c'$ on $\T'$ such that $L(\T',c')  \equiv_D L(\T,c)$, and construct such a $c'$ with the smallest possible number of priorities if one exists.

\begin{theorem} \label{thm:prio-opt}
The problem {\gendcpriorityopt} can be solved in polynomial time.
\end{theorem}

The proof of this theorem follows from Lemma~\ref{lem:consistent-parity}, Lemma~\ref{lem:poly-subset-priority-oracle}, and Theorem~\ref{thm:subset-priority-oracle} below. We start with some definitions.

\medskip
Let $\T = (\Sigma,Q,\delta,q_0)$ and $\T' = (\Sigma,Q',\delta',q_0')$, and for $P \subseteq Q$ define $h(P) = \{h(q) \mid q \in P\}$. From the DPA $\A$, the don't care set $D$, and the homomorphism $h$, we obtain two families $(\F_0',\F_1')$ of subsets of $Q'$, those that need to be accepting, and those that need to be rejecting in order to accept a language that is $D$-equivalent to $L(\A)$. Formally, we define
\[
\begin{array}{l}
  \F_0' = \{R' \subseteq Q' \mid \exists \alpha \in L(\A) \setminus D\;:\; h(\infs[\A](\alpha)) = R'\} \\
  \F_1' = \{R' \subseteq Q' \mid \exists \alpha \in \comp{L(\A)} \setminus D\;:\; h(\infs[\A](\alpha)) = R'\}
\end{array}
\]
We say that a parity condition $c':Q' \rightarrow \nat$ is consistent with $(\F_0',\F_1')$ if $\max(c'(R'))$ is even for all $R' \in \F_0'$, and odd for  all $R' \in \F_1'$. See Figure~\ref{fig:prio-opt} for an example.
\begin{lemma} \label{lem:consistent-parity}
With the notations introduced above, $L(\T',c')  \equiv_D L(\T,c)$ iff $c'$ is consistent with $(\F_0',\F_1')$.
\end{lemma}

\begin{proof}
  The proof is immediate from the definitions: Let $\alpha \in \Sigma^\omega \setminus D$. Since $h$ is a homomorphism, $\infs[\T'](\alpha) = h(\infs(\alpha))$. Thus, $\infs[\T'](\alpha) \in \F_0'$ if $\alpha \in L(\A)$, and $\infs[\T'](\alpha) \in \F_1'$ if $\alpha \in \comp{L(\A)}$. So $L(\T',c')  \equiv_D L(\T,c)$ iff $(\T',c')$ accepts precisely those words in $\Sigma^\omega \setminus D$
with infinity set in $\F_0'$, which is the definition of consistency with $(\F_0',\F_1')$.
\end{proof}
So for solving the problem {\gendcpriorityopt}, we need to find an optimal parity condition $c'$ that is consistent with $(\F_0',\F_1')$. There is an algorithm that solves this problem in polynomial time, which is presented in \cite{BohnL21}. Unfortunately, the size of $(\F_0',\F_1')$ can be exponential in general, so we cannot construct it explicitly in a polynomial time algorithm. However, the algorithm from \cite{BohnL21} only needs to access $(\F_0',\F_1')$ via queries of the following type: Given a set $P' \subseteq Q'$ and $t \in \{0,1\}$, return the set
\[
U_t(P') = \bigcup \{R' \subseteq P' \mid R' \in \F_t'\},
\]
that is, the union of all sets $R' \in \F_t'$ with $R' \subseteq P'$. Let us call such a function that returns $U_t(P')$  a \textit{Subset-Parity-Oracle} for $(\F_0',\F_1')$. We can thus reformulate the result from \cite{BohnL21} as follows.

\begin{theorem}[{\cite[Theorem 4]{BohnL21}}] \label{thm:subset-priority-oracle}
There is a polynomial time algorithm that, given a \textit{Subset-Parity-Oracle} for $(\F_0',\F_1')$ that is computable in polynomial time, constructs an optimal parity condition $c'$ on $\T'$ that is consistent with $(\F_0',\F_1')$ if one exists.
\end{theorem}

Thus, for a polynomial time algorithm solving {\gendcpriorityopt}, it suffices to show that one can implement a \textit{Subset-Parity-Oracle} for $(\F_0',\F_1')$ in polynomial time, which is stated in the following lemma.
\begin{lemma} \label{lem:poly-subset-priority-oracle}
There is a polynomial time algorithm that, given the inputs $\A = (\T,c),\D,\T',h$ for the problem {\gendcpriorityopt}, a set $P'$ of states of $\T'$, and $t \in \{0,1\}$, returns the set $U_t(P')$.
\end{lemma}

\begin{proof}
Let $\A=(\T,c)$ with $\T = (\Sigma,Q,\delta,q_0)$, $\D= (\Sigma,Q_\D,\delta_\D,q_0^\D,d)$, and $\T' = (\Sigma,Q',\delta',q_0')$.
We consider the case $t = 0$, the case $t=1$ being symmetric. Let $R' \subseteq P'$ with $R' \in \F_0'$. By definition of $\F_0'$, there is $\alpha \in L(\A) \setminus D$ with  $h(\infs[\A](\alpha)) = R'$. For the condition $\alpha \in L(\A) \setminus D$, we consider the product of $\A$ and $\D$, which is obtained by the standard product construction, restricted to reachable states, and with the two priority functions $c,d$ extended to pairs of states. Formally we define $\hat{\A} = \A \times \D = (\Sigma, \hat{Q}, \hat{\delta},\hat{q_0},\hat{c},\hat{d})$ with
\begin{itemize}
\item $\hat{q_0} = (q_0,q_0^\D)$,
\item $\hat{Q} = Q \times Q_\D$ restricted to the states reachable from $\hat{q_0}$ by application of $\hat{\delta}$,
\item $\hat{\delta}((q_1,q_2),a) = (\delta(q_1,a),\delta_\D(q_2,a))$
\item $\hat{c}(q_1,q_2) = c(q_1)$ and $\hat{d}(q_1,q_2) = d(q_2)$
\end{itemize}
Then $\alpha \in L(\A) \setminus D$ is equivalent to $\max(\hat{c}(\infs[\hat{\A}](\alpha)))$ being even and $\max(\hat{d}(\infs[\hat{\A}](\alpha)))$ being odd. This means that the set $\infs[\hat{\A}](\alpha)$ is a strongly connected set in $\hat{\A}$ on which the maximal $c$-priority is even, and the maximal $d$-priority is odd.

We can view $h$ also as a homomorphism from the product transition system $\hat{\A}$ to $\T'$ by setting $h(q_1,q_2) = h(q_1)$ for all $(q_1,q_2) \in \hat{Q}$. The condition $h(\infs[\A](\alpha)) = R'$ then directly translates to $\hat{\A}$ as $h(\infs[\hat{\A}](\alpha)) = R'$.

In summary, the sets $R' \in \F_0$ with $R' \subseteq P'$ are precisely those sets $h(S)$ where $S$ is a strongly connected subset of states in $\hat{\A}$ with $h(S) \subseteq P'$, $\max(\hat{c}(S))$ even, and $\max(\hat{d}(S))$ odd.

\medskip
Hence, for computing $U_0(P')$, we can proceed as follows.
\begin{itemize}
\item $U_0 := \emptyset$
\item For all even $i$ in the range of $c$, and all odd $j$ in the range of $d$:
  \begin{itemize}
  \item Let $\hat{\A}_{P',i,j}$ be $\hat{\A}$ restricted to the states $(q_1,q_2) \in \hat{Q}$ with $h(q_1) \in P'$, $c(q_1) \le i$, and $d(q_2) \le j$.
  \item For each MSCC $S$ of $\hat{\A}_{P',i,j}$ with $\max(\hat{c}(S)) = i$ and $\max(\hat{d}(S)) = j$, let $U_0 := U_0 \cup h(S)$.
  \end{itemize}
\item Return $U_0$
\end{itemize}
This procedure can be implemented to run in polynomial time because there are only polynomially many pairs $i,j$ to consider, the product $\hat{\A}$ and its restrictions $\hat{\A}_{P',i,j}$ as well as the MSCCs can all be computed in polynomial time.
\end{proof}

\section{State space minimization}\label{sec:state-minimization}

In this section we consider minimization problems for deterministic $\omega$-automata with informative right-congruence under a given don't care set. As explained in the introduction,
we only consider don't care sets $D$ with trivial right-congruence.
We start with the observation that for such don't care sets, one can define a canonical right-congruence as a generalization of the standard Myhill-Nerode congruence.

In the following, if not stated otherwise, let $L \subseteq \Sigma^\omega$ be a regular $\omega$-language and $D \subseteq \Sigma^\omega$ a regular $\omega$-language of don't care words with trivial right-congruence.

\begin{definition}[$D$-congruence]\label{def:d-con}
	Given $w,w' \in \Sigma^*$, we say that $w$ and $w'$ are $D$-equivalent for $L$, denoted as $w \sim_{L,D} w'$, if and only if for all $\alpha \in \Sigma^\omega \setminus D$ it holds that $w\alpha \in L \iff w'\alpha \in L$. Note that $\sim_{L,\emptyset}$ is the same as $\sim_L$.
\end{definition}
Clearly, $w \sim_{L} w'$ implies $w \sim_{L,D} w'$ for any don't care set $D$. The opposite is obviously not true because for $D=\Sigma^\omega$, all finite words are $D$-equivalent for every $L$.
If $D$ has trivial right-congruence, it follows that $\alpha \in D \iff w\alpha \in D$ for all $\alpha \in \Sigma^\omega$ and $w \in \Sigma^*$. One can conclude that $\sim_{L,D}$ is a right-congruence, as stated in the following proposition.

\begin{proposition}\label{peq}
 For $L,D \subseteq \Sigma^\omega$ such that $D$ has trivial right-congruence, $\sim_{L,D}$ is a right-congruence.
\end{proposition}

\begin{proof}
That $\sim_{L,D}$ is an equivalence relation directly follows from the definition and does not require $D$ to have trivial right-congruence. Now assume that $u \sim_{L,D} v$ and let $\sigma \in \Sigma$. For every $\alpha \in \Sigma^\omega \setminus D$ we also have $\sigma \alpha \in \Sigma^\omega \setminus D$ since $D$ has trivial right-congruence. Hence $u\sigma\alpha \in L \Leftrightarrow v\sigma\alpha \in L$.
\end{proof}

\begin{definition}[Informative $D$-congruence]\label{def:inf-d-con}
Let $\TLD$ be the transition system associated with $\simLD$.
	A deterministic $\omega$-automaton $\mathcal{A}$ with $L:=L(\mathcal{A})$ possesses informative $D$-congruence if its transition system is isomorphic to $\TLD$. We also define $\IB(D)$ as the class of those $\omega$-languages $L$ for which there is a DBA $\A$ with informative $D$-congruence such that $L(\A) \equiv_D L$. The corresponding definitions for DCA, DPA, and DMA define the classes $\IC(D),\IP(D),\IM(D)$. Note that for $D=\emptyset$ we obtain the definitions of informative right-congruence.
\end{definition}

If $\A$ is a DMA with $L(\A) \equiv_D L$, then clearly two words $w,w' \in \Sigma^*$ that lead to the same state are $\simLD$-equivalent. This leads to the observation that a DMA with informative $D$-congruence is $D$-minimal (since DPA, DBA etc.~can be viewed as DMA, this observation holds for all acceptance types considered in this paper). Thus, the classes $\IB(D), \IC(D),\IP(D),\IM(D)$ consist of those languages that have a $D$-minimal automaton with informative $D$-congruence (with the respective acceptance condition).

Based on the results from Section~\ref{sec:priority-optimization}, we can show that it is decidable in polynomial time whether a given DPA has a $D$-minimal automaton with informative $D$-congruence.
\begin{theorem} \label{thm:decide-IRCD}
For given DPAs $\A$ and $\D$ where $D = L(\D)$ has a trivial right-congruence, it is decidable in polynomial time whether $L(\A)$ is in $\IB(D),\IC(D)$, or $\IP(D)$, and a corresponding DBA, DCA, or DPA $\B$ with informative $D$-congruence and $L(\B) \equiv_D  L(\A)$ can be constructed in polynomial time if one exists.
\end{theorem}

 \begin{proof}
  Let $L = L(\A)$. Given two states $p,q$ of $\A$, it can be decided in polynomial time whether $p \not\sim_{L,D} q$: One can build the product $\hat{\A}$ of $\A \times \A \times \D$, and then check if there is a strongly connected set $S$ in this product graph that is reachable from $(p,q,q_0^\D)$ (where $q_0^\D$ is the initial state of $\D$) and on which
  \begin{enumerate}[(a)]
  \item the maximal priority of states in the first component has different parity than the maximal priority of states in the second component (the resulting word is accepted from $p$ iff it is not accepted from~$q$),
  \item the maximal priority of states in the third component is odd (the resulting word is not a don't care word).
  \end{enumerate}
  An $\omega$-word that reaches $S$ from $(p,q,q_0^\D)$ and has $S$ as infinity set, is a word outside $D$ that witnesses the difference between $p$ and  $q$.

\medskip
  The existence of of such a set can be tested by first removing all states in $\hat{\A}$ that are not reachable from  $(p,q,q_0^\D)$, leading to $\hat{\A}'$. Then one iterates through all triples $(i,j,k)$ of priorities that satisfy properties (a) and (b). For each such triple one removes all states from $\hat{\A}'$ with priority triple $(i',j',k')$ and $i' > i$ or $j'>j$ or $k' > k$, and then computes the MSCCs in the remaining graph. If one finds an MSCC with a state that has priority triple $(i,j,k)$, then one has found a witness for $p \not\sim_{L,D} q$.
If this procedure does not find a witness for  $p \not\sim_{L,D} q$ for any of the considered triples $(i,j,k)$, then the states are equivalent.

  Once $\sim_{L,D}$ is computed, one can construct $\TLD$ and the corresponding homomorphism from the transition structure of $\A$ onto $\TLD$. By Theorem~\ref{thm:prio-opt} one can decide in polynomial time whether there is a parity condition on $\TLD$ such that the resulting language is $D$-equivalent to $L(\A)$.

  Since furthermore the constructed parity condition is optimal, one can also deduce whether there is a B\"uchi or co-B\"uchi condition on $\TLD$ that leads to a language that is $D$-equivalent to $L(\A)$.
\end{proof}

For $D=\emptyset$, this solves the minimization problem for IRC languages and DPAs, which has independently been obtained in \cite[Theorem 15.2]{ASF22} with a different proof.
\begin{corollary}
  Given a DPA $\A$, it is decidable in polynomial time whether $L(\A)$ is in $\IB,\IC$, or $\IP$, and a corresponding DBA, DCA, or DPA with informative right-congruence for $L(\A)$ can be constructed in polynomial time if one exists.
\end{corollary}

We have seen that one can efficiently check if a language of a DPA has a $D$-minimal automaton with informative $D$-congruence. We now analyze the question for which language classes $D$-minimal automata are guaranteed to have informative $D$-congruence.
As a first observation, we can explain the result from \cite{EisingerK06} that for each WDBAs there is a unique $D$-minimal WDBA in terms of the $D$-congruence.
\begin{theorem}\label{t1}
	Given a WDBA $\mathcal{A}=(\Sigma,Q,\delta,q_0,F)$ with $L:=L(\mathcal{A})$ and a regular $\omega$-language of don't care words $D \subseteq \Sigma^\omega$ with trivial right-congruence, $\mathcal{A}$ is $D$-minimal if and only if it has informative $D$-congruence.
\end{theorem}
\begin{proof}
Clearly, for any $u,v \in \Sigma^*$, it holds that $\delta^*(q_0,u)=\delta^*(q_0,v)$ implies $u \sim_{L,D} v$. It follows then that $\mathcal{A}$ has at least $Ind(\sim_{L,D})$ states. Thus, we can conclude that any WDBA with informative $D$-congruence is $D$-minimal, as well.

\medskip
We now show that for each equivalence class $\kappa$ of $\sim_{L,D}$, all $\alpha \in \Sigma^\omega \setminus D$ with infinitely many prefixes in $\kappa$ are in $L$, or all these $\alpha$ are outside $L$. This defines an acceptance status for each equivalence class such that for two equivalence classes that are strongly connected (in the transition system $\TLD$) the acceptance status is the same. Hence, $\TLD$ can be equipped with a weak acceptance condition such that the accepted language is $D$-equivalent to $L$.

As mentioned above, all finite words leading to the same state $q$ are in the same equivalence class of $\sim_{L,D}$. If this equivalence class is $\kappa$, we call $q$ a $\kappa$-state.

Let $\alpha,\beta \in \Sigma^\omega \setminus D$ be such that they both have infinitely many prefixes in $\kappa$. Assume that $\alpha \in L$. We show that also $\beta \in L$ (the other case is symmetric). Let $\alpha = u\alpha'$ and $\beta = v\beta'$, with $u,v \in \kappa$. Let $q$ be a $\kappa$-state in $\A$ such that no other $\kappa$-state in a different MSCC is reachable from $q$, and let $w \in \Sigma^*$ such that $\delta^*(w) = q$. Then $w\alpha' \in L \setminus D$ because $\alpha' \notin D$ and $u \sim_{L,D} w$.

Since $u\alpha'$ and $v\beta'$ have infinitely many prefixes in $\kappa$, the same is true for $w\alpha'$ and $w\beta'$. Thus $\infs(w\alpha'),\infs(w\beta') \subseteq \text{MSCC}(q)$, by choice of $q$. Since $w \alpha' \in L$ as seen above, also $w\beta' \in L$ because $\A$ is a WDBA. We conclude that $\beta = v\beta' \in L$ because $v \sim_{L,D} w$.
\end{proof}

This result for WDBAs raises the question whether it can be extended to larger classes of automata. The next natural candidate is the class of deterministic $\omega$-automata with informative right-congruence. Since these have unique transition graphs defined by the right-congruence of their language, one might hope that they behave similarly to WDBAs under don't care sets with trivial right-congruence. We show that this is not the case. Furthermore, $D$-minimal automata for deterministic $\omega$-automata with informative right-congruence need not to be unique and do not necessarily have informative right-congruence.

\begin{proposition}\label{obs_not_uniqe}
  There is a don't care set $D$ with trivial right-congruence, and a DBA $\A$ with informative right-congruence, such that there are several non-isomorphic $D$-minimal DBA for $\A$, and a $D$-minimal DBA for $\A$ does not necessarily have informative right-congruence.
\end{proposition}

\begin{figure}[!h]
\vspace*{-4mm}
	\centering
	\begin{tikzpicture} [thick,scale=0.7, every node/.style={scale=0.7}]
		\node (q0) [state, initial, initial text = {$\mathcal{A}:$}] {$q_0$};
		\node (q1) [state, accepting, right = of q0] {$q_1$};
		\node (q2) [state, below = of q1] {$q_2$};
		
		\path [-stealth, thick]
		(q0) edge [bend left] node [above=0.1cm] {$b,c$} (q1)
		(q0) edge [bend right] node [left=0.1cm] {$a$} (q2)
		(q1) edge [bend left] node [below=0.1cm] {$b$} (q0)
		(q1) edge node [left=0.1cm] {$a,c$} (q2)
		(q2) edge [bend right] node [right=0.1cm] {$c$} (q1)
		(q2) edge [loop right] node {$a,b$}();
	\end{tikzpicture}
	\begin{tikzpicture}[thick,scale=0.7, every node/.style={scale=0.7}] 
		\node (q0) [state, initial, initial text = {$\mathcal{B}:$}] {$q_0$};
		\node (q1) [state, accepting, right = of q0] {$q_1$};
		
		\path [-stealth, thick]
		(q0) edge [loop above] node {$a,b$}()
		(q0) edge [bend left] node [above=0.1cm] {$c$} (q1)
		(q1) edge [bend left] node [below=0.1cm] {$a,b,c$} (q0);
	\end{tikzpicture}\hspace{.22cm}
	\begin{tikzpicture}[thick,scale=0.7, every node/.style={scale=0.7}] 
		\node (q0) [state, initial, initial text = {$\mathcal{C}:$}] {$q_0$};
		\node (q1) [state, accepting, right = of q0] {$q_1$};
		
		\path [-stealth, thick]
		(q0) edge [loop above] node {$a,b$}()
		(q0) edge [bend left] node [above=0.1cm] {$c$} (q1)
		(q1) edge [bend left] node [below=0.1cm] {$a,c$} (q0)
		(q1) edge [loop above] node {$b$}();
	\end{tikzpicture}\hspace{.2cm}
	\begin{tikzpicture}[thick,scale=0.7, every node/.style={scale=0.7}] 
		\node (q0) [state, initial, initial text = {$\mathcal{D}:$}] {$q_0$};
		\node (q1) [state, accepting, right = of q0] {$q_1$};
		
		\path [-stealth, thick]
		(q0) edge [loop above] node {$a,b$}()
		(q0) edge [bend left] node [above=0.1cm] {$c$} (q1)
		(q1) edge [bend left] node [below=0.1cm] {$a$} (q0)
		(q1) edge [loop above] node {$b,c$}();
	\end{tikzpicture}\vspace*{-1mm}
	\caption{The DBA $\mathcal{A}$ has 3 different corresponding $D$-minimal DBA for the set of \textit{don't care words} $D:=\Sigma^*b^\omega$.}
	\label{fig:IRCA}\vspace*{-5mm}
\end{figure}
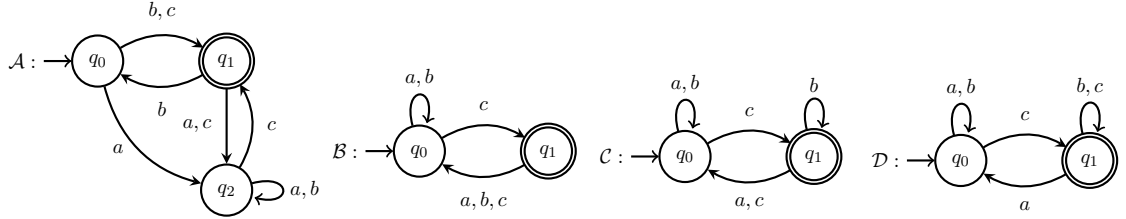

\begin{proof}
  Consider the DBA $\mathcal{A}$ in Figure~\ref{fig:IRCA} and let $L:=L(\mathcal{A})$. Each state corresponds to a different equivalence-class in $\sim_L$, namely $q_0$ to $[\epsilon]$, $q_1$ to $[b]$ and $q_2$ to $[a]$. These equivalence classes are distinct since $b^\omega \in L$ and  $ab^\omega \not\in L$, showing $\epsilon \not\sim_L a$ and $b \not\sim_L a$, and $\epsilon \not\sim_L b$ is witnessed by $cb^\omega \in L$ and $bcb^\omega \not\in L$. Thus, $L(\mathcal{A}) \in \mathbb{IB}$ and $\mathcal{A}$ is minimal for $L$. With $D=\Sigma^*b^\omega$, we have a set of \textit{don't care words} with trivial right-congruence, and any $D$-minimal DBA for $L$ has at least two states as $a^\omega \not\in L(\mathcal{A}) \cup D$ and $c^\omega \in L(\mathcal{A}) \backslash D$. Figure~\ref{fig:IRCA} shows three different $D$-minimal DBA $\mathcal{B}$, $\mathcal{C}$, and $\mathcal{D}$, all with different transition graphs. Note that $L(\mathcal{C}), L(\mathcal{D}) \in \mathbb{IB}$ as $b^\omega \not \sim_{L'} cb^\omega$ for $L' \in \{L(\mathcal{C}), L(\mathcal{D})\}$, but $L(\mathcal{B})$ is the set of all words containing infinitely many $c$, which has trivial right-congruence, so $\B$ does not have informative right-congruence.

\medskip
We used the set  $D=\Sigma^*b^\omega \in \mathbb{DC}$ of don't care words, but the same results hold for the set of don't care words $D'=(\Sigma^*b)^\omega \in \mathbb{DB}$.
\end{proof}

We now consider the computational complexity of the $D$-minimization problem for automata with IRC, and show that it is NP-hard.
\eject

\begin{definition}[$D$-minimization Problem for IRC]\label{def:d-min_w_irc}
	Given a deterministic $\omega$-automaton $\mathcal{A}$ with informative right-congruence, a DPA $\D$ accepting a set of don't care words $D = L(\D)$ with trivial right-congruence, and $k \in \mathbb{N}$, does a corresponding $D$-minimal $\omega$-automaton $\mathcal{B}$ of the same type as $\A$ with at most $k$ states exist?
\end{definition}

\begin{theorem}\label{thm:NPhard}
	The $D$-minimization Problem for IRC is NP-hard for DBA, DCA, and DPA.
\end{theorem}

We show NP-hardness via reduction from the NP-complete Vertex Coloring Problem \cite{Karp1972}:
	Given an undirected graph $\mathcal{G}=(V,E)$ with $|E|>0$ and $k \in \mathbb{N}$, does there exist a coloring $col:V \rightarrow \{1, \ldots, k\}$ of the vertices such that $col(v) \not= col(u)$ for all $(u,v) \in E$?

\medskip
We first briefly describe the idea for the reduction, and then enter the formal details.
For an undirected graph $\mathcal{G}=(V,E)$ with $|E|>0$, we take the vertex set as alphabet, and consider the language of all $\alpha \in V^\omega$ in which only one vertex appears infinitely often. This language can be accepted by a DCA (a DPA with priorities $0$ and $1$) with one state of priority $0$ for each vertex $v$, which is reached with and loops on $v$, and one extra state $q_G$ of priority $1$ that is reached whenever the next vertex in the sequence is different from the previous one. As an example, consider the DCA $\A_G$ in Figure~\ref{fig:AG} without the transition labels of the form $x_v$. These labels are added to ensure that the resulting language has informative right-congruence: the sequence $x_v^\omega$ is accepted only from the state for $v$.

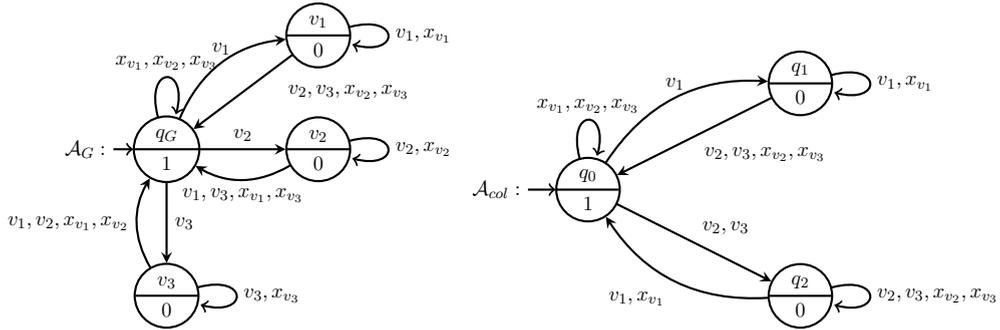
\begin{figure}[!ht]
\vspace*{-1mm}
	\centering
	\begin{tikzpicture}[thick,scale=0.5, every node/.style={scale=0.7}]
		\node (q0) [state with output, initial, initial text = {$\mathcal{A}_G:$}] at (0,-3) {
			$q_G$ \nodepart{lower} $1$
		};
		\node (1) [state with output] at (4,0) {
			$v_1$ \nodepart{lower} $0$
		};
		
		\node (2) [state with output] at (4,-3) {
			$v_2$ \nodepart{lower} $0$
		};
		
		\node (3) [state with output, below = of 2] at (0,-4) {
			$v_3$ \nodepart{lower} $0$
		};

		\path [-stealth, thick]
		(q0) edge [loop above] node {$x_{v_1},x_{v_2},x_{v_3}$}()
		(q0) edge [bend left] node [above=0.01cm] {$v_1$} (1)
		(q0) edge node [above=0.01cm] {$v_2$} (2)
		(q0) edge node [right=0.01cm] {$v_3$} (3)
		(1) edge node [right=0.5cm] {$v_2,v_3,x_{v_2},x_{v_3}$} (q0)
		(1) edge [loop right]  node {$v_1,x_{v_1}$}()
		(2) edge [bend left] node [below=0.01cm] {$v_1,v_3,x_{v_1},x_{v_3}$} (q0)
		(2) edge [loop right]  node {$v_2,x_{v_2}$}()
		(3) edge [bend left] node [left=0.01cm] {$v_1,v_2,x_{v_1},x_{v_2}$} (q0)
		(3) edge [loop right]  node {$v_3,x_{v_3}$}();
	\end{tikzpicture}
	\begin{tikzpicture}[thick,scale=0.7, every node/.style={scale=0.7}]
		\node (q0) [state with output, initial, initial text = {$\mathcal{A}_{col}:$}] at (0,0){
			$q_0$ \nodepart{lower} $1$
		};
		\node (1) [state with output] at (4,2) {
			$q_1$ \nodepart{lower} $0$
		};
		\node (2) [state with output]  at (4,-2) {
			$q_2$ \nodepart{lower} $0$
		};
		
		\path [-stealth, thick]
		(q0) edge [loop above] node {$x_{v_1},x_{v_2},x_{v_3}$}()
		(q0) edge [bend left] node [above=0.01cm] {$v_1$} (1)
		(q0) edge node [above right=0.01cm] {$v_2,v_3$} (2)
		(1) edge node [below right=0.05cm] {$v_2,v_3,x_{v_2},x_{v_3}$} (q0)
		(1) edge [loop right]  node {$v_1,x_{v_1}$}()
		(2) edge [bend left] node [below left=0.05cm] {$v_1,x_{v_1}$} (q0)
		(2) edge [loop right]  node {$v_2,v_3,x_{v_2},x_{v_3}$}();
	\end{tikzpicture}
	\caption{The $D'$-minimal DPA $\mathcal{A}_{col}$ and the DPA $\mathcal{A}_G$ for  $\mathcal{G}=(\{v_1,v_2,v_3\},\{(v_1,v_2),(v_1,v_3)\})$ recognize $D'$-equivalent $\omega$-languages.}
	\label{fig:AG}
\end{figure}

The don't care set then includes all words $\alpha$ in which labels of the form $x_v$ occur infinitely often, or from some point on successive vertices $vv'$ in $\alpha$ are not connected by an edge. With this don't care set, a $D$-minimal DPA $\A_{\textit{col}}$ can use one state for each color class in a $k$-coloring instead of one state for each vertex $v$, as illustrated in Figure~\ref{fig:AG}. This shows NP-hardness of $D$-minimization for DCA and DPA with IRC, and can easily be adapted to also work for DBA.

\medskip
We now give the formal details of the reduction. For an undirected graph $G$ with at least one edge, define the DPA $\mathcal{A}_G:=(\Sigma,V\cup \{q_0\},\delta_G,q_G,c_G)$ with $\Sigma:=V \cup \{x_v \mid v \in V\}$ and for all $u,v \in V, u \neq v$:
	\begin{itemize}
		\item $\delta_G(q_G,u)=u$ and $\delta_G(q_G,x_u)=q_G$,
		\item $\delta_G(u,u)=u$, $\delta_G(u,x_u)=u$, $\delta_G(u,v)=q_G$ and $\delta_G(u,x_v)=q_G$,
\eject
		\item $c_G(q_G)=1$ and $c_G(u)=0$.
	\end{itemize}
Furthermore, define the set of don't care words as $D := D' \cup D''$ with
	\begin{itemize}
	\item $D':=\{w\alpha \mid w \in \Sigma^*, \alpha \in V^\omega, \forall i \in \mathbb{N}: \exists j>i: \alpha(i)\neq\alpha(j) \text{ and } \forall i \in \mathbb{N}:\alpha(i)\neq\alpha(i+1) \rightarrow (\alpha(i),\alpha(i+1))\not \in E\}$
	\item $D'':=\{\alpha \in \Sigma^\omega \mid \exists v \in V: \alpha(i)=x_v \text{ for infinitely many } i\in\mathbb{N}\}.$
        \end{itemize}
        This set $D$ is accepted by the DPA  $(\Sigma,\{q_x\} \cup (V \times \{1,2,3\}), \delta_D,q_x,c_D)$ with
        \[
        \delta_D((u,i),v):=
        \begin{cases}
	(v,1), & \text{if } u=v\\
	(v,2), & \text{if } u \neq v \text{ and } (u,v) \not \in E \\
	(v,3), & \text{if } (u,v) \in E\\
        \end{cases}
        \]
for all $u,v \in V$ and $i \in \{1,2,3\}$, as well as $\delta_D(q_x,v):=(v,2)$ and $\delta_D(q,x_v)=q_x$ for $q \in \{q_x\} \cup (V \times \{1,2,3\})$ and $v \in V$, and finally $c_D(v,i)=i$ and $c_D(q_x)=4$.

\begin{lemma}
	$\mathcal{A}_G$ has informative right-congruence.
\end{lemma}
\begin{proof}
	Let $L:=L(\mathcal{A}_G)$. For every $u,v \in V$ with $u \neq v$, we have $ux_u^\omega \in L$ and $x_u^\omega, vx_u^\omega \not \in L$. Consequently, for every $u,v \in V$ with $u \neq v$, it holds that $\epsilon \not \sim_L u$ and $u \not \sim_L v$.
\end{proof}

\begin{lemma} \label{lem:colorG2}
  $G$ can be colored using at most $k$ different colors if, and only if, any $D$-minimal DPA $\mathcal{B}$ with $L(\mathcal{A}_G) \equiv_{D} L(\mathcal{B})$ has at most $k+1$ states.
\end{lemma}

\begin{proof}
	First, we demonstrate that if $G$ can be colored using at most $k$ colors, this implies that any $D$-minimal DPA $\mathcal{B}$ with $L(\mathcal{A}_G) \equiv_{D} L(\mathcal{B})$ has at most $k+1$ states: We assume that a vertex-coloring $col:V \rightarrow \{1,\dots,k\}$ on $G$ exists. Then we construct a corresponding DPA $\mathcal{A}_{col}:=(\Sigma,\{q_0,\dots,q_k\},\delta,q_0,c)$ with transitions $\delta(q_0,v):=q_{col(v)}$, $\delta(q_0,x_v):=q_0$, 	
	\[\delta(q_i,v):= \begin{cases}
		q_i, & \text{if } i=col(v) \\
		q_0, & \text{if } i\neq col(v)\text{,}
	\end{cases} \text{ and }
	\delta(q_i,x_v):= \begin{cases}
		q_i, & \text{if } i=col(v) \\
		q_0, & \text{if } i\neq col(v)\text{,}
	\end{cases}\]
	as well as $c(q_0)=1$ and $c(q_i)=0$ for all $i \in \{1,\dots,k\}$ and $v \in V$. Since $col$ is a valid vertex-coloring, it follows that $L(\mathcal{A}_{col})$ only differs from $L(\mathcal{A}_G)$ on words that are in $D$, and hence $L(\mathcal{A}_G) \equiv_{D} L(\mathcal{A}_{col})$. 

 \medskip Next we show that if a $D$-minimal DPA for $L(\mathcal{A}_G)$ has at most $k+1$ states, then $G$ can be colored using at most $k$ colors: Assume that there exists a DPA $\mathcal{C}:=(\Sigma,\{0,\dots,k\},\delta,0,c)$ with $L(\mathcal{A}_G) \equiv_{D} L(\mathcal{C})$. There is at least one state $q \in \{0,\dots,k\}$ such that $c(q)$ is odd. W.l.o.g.~let $c(0)$ be odd. Since $\mathcal{C}$ needs to accept each $v^\omega$, we can define
        \[
        col:V \rightarrow \{1,\dots,k\} \text{ with } col(v):=\min\{q \in \infs[\mathcal{C}](v^\omega) \mid c(q) \text{ is even}\} \text{ for } v \in V.
        \]
        For $(u,v) \in E$ any $\omega$-word $\alpha \in (V^*uv)^\omega$ is rejected by $\mathcal{C}$. Thus, we have $\infs[\mathcal{C}](u^\omega) \cap \infs[\mathcal{C}](v^\omega)=\emptyset$. This mean that $col$ is a valid vertex coloring on $G$ which uses at most $k$ different colors.
\end{proof}

  As demonstrated in Lemma~\ref{lem:colorG2}, we can reduce any instance of the Vertex Coloring Problem to an instance of the $D$-minimization Problem for IRC for DPA. Furthermore, the reduction is possible in polynomial time, as the alphabet has a size of $|\Sigma|=2|V|$ and the number of states in $\mathcal{A}_G$ is $|V|+1$. It follows then that the $D$-minimization Problem with IRC for DPA is NP-hard.

  The DPA used in the reduction are actually DCA, as we only utilize priorities $0$ (final state) and $1$ (non-final state). Consequently, the $D$-minimization Problem with IRC is NP-hard for DBA and DCA, as well.
This finishes the proof of Theorem~\ref{thm:NPhard}.

\section{Active learning of WDBA}\label{sec:learning}

In \cite{MALER1995316} it is shown that the minimal WDBA for an unknown language $U$ that is WDBA recognizable can be learned in polynomial time from membership and equivalence queries (the time is polynomial in the size of the minimal WDBA for $U$ and the size of the counterexamples provided by the oracle). In this section, we consider the extension of this setting with an additional set $D$ of don't care words with trivial right-congruence that is also unknown to the learner. We only consider WDBA in this context, because this is the only class of $\omega$-automata for which a polynomial time active learning algorithm (with standard membership and equivalence queries) is known. The algorithm in \cite{AngluinF16} is not polynomial, and learns a different representation whose translation to deterministic $\omega$-automata is exponential, in the worst case. Furthermore, in \cite{BohnL21} it is shown that the IRC restriction does not help for active learning, in the sense that polynomial time active learning of regular $\omega$-languages with IRC is not easier than polynomial time active learning of deterministic $\omega$-automata in general.

\medskip
As in \cite{MALER1995316}, our algorithm can use membership and equivalence queries. Since there are don't care words, a membership query $\MQ(u,v)$ for $u \in \Sigma^*$ and $v \in \Sigma^+$ returns ``yes'' if $uv^\omega \in U\setminus D$, ``no'' if $uv^\omega \not\in U\cup D$ and ``don't care'' if $uv^\omega \in D$. For a hypothesis WDBA $\Hy$, an equivalence query $\EQ(\Hy)$ either confirms that $L(\mathcal{H}) \equiv_D U$, or provides a counterexample, which is an ultimately periodic word $uv^\omega \in \Sigma^\omega \setminus D$ with $uv^\omega \in (L(\mathcal{H})\backslash U)\cup(U \backslash L(\mathcal{H}))$.

It turns out that it suffices to modify one part of the algorithm, which is called ``conflict resolution'' in \cite{MALER1995316}, in order to deal with the don't care words. As in the original algorithm, we use an observation table in order to approximate $\sim_{U,D}$. An observation table consists of a prefix closed set $S \subseteq \Sigma^*$, a suffix closed set $E$ of ultimately periodic words with $E \cap D = \emptyset$, and a  function $f:(S \cup S\Sigma) \times E \rightarrow \{\text{``yes'',''no''}\}$ with $f(s,uv^\omega):=\MQ(su,v)$ for all $s \in S$ and $uv^\omega \in E$. Since $D$ has trivial right-congruence and $D \cap E = \emptyset$, we know that $\MQ(su,v)$ is ``yes'' or ``no'' for all $s \in S$ and $uv^\omega \in E$, so we will never ask membership queries for don't care words.

In the table view, the rows are indexed with the elements from $S \cup S\Sigma$ and the columns are indexed with the elements from $E$. The entries in the row for $s \in S \cup S\Sigma$ are an approximation of the information required for identifying the equivalence class of $s$ in $\sim_{L,D}$. This row is captured by the function $f_s: E \rightarrow \{\text{``yes'',''no''}\}$ with $f_s(uv^\omega):=f(s,uv^\omega)$ for each $s \in S \cup S\Sigma$. The observation tables are constructed in such a way that the rows of all elements in $S$ are different, a property that is often referred to as reduced observation table. In the following we only work with reduced observation tables.

The elements from $S$ correspond to the states of the WDBA, and the rows indexed by $S\Sigma$ are used for the transitions.
We say that an observation table is closed if and only if for every $s \in S\Sigma$ there is $t \in S$ with $f_s=f_t$.
For a closed observation table we can then define the corresponding transition system $\TSf = (\Sigma,S,\delta_f,\epsilon)$ with $\delta_f(s,\sigma):=t$ for the unique $t \in S$ with $f_{s\sigma}=f_t$.  The following lemma states that this transition system is isomorphic to the transition system induced by $\sim_{U,D}$ if $S$ contains as many elements as there are equivalence classes.

\begin{lemma}\label{lem:isomorphic}
If $|S| = \ind(\sim_{U,D})$, then $\TSf$ is isomorphic to $\TUD$.
\end{lemma}
\begin{proof}
Since $|S| = \ind(\sim_{U,D})$, there is one representative of each equivalence class in $S$. For $s \in S$ and $\sigma \in \Sigma$, the unique $t \in S$ with $f_{s\sigma}=f_t$ must be the representative of the equivalence class of $s\sigma$.
\end{proof}

The following lemma states that for $\omega$-words that belong to the table, the transition system is always in a state that classifies the remaining suffix correctly.
\begin{lemma} \label{lem:consistent}
Let $(S,E,f)$ be a closed observation table, $\alpha \in S\cdot E$, and $w \in\Sigma^*$ be a prefix of $\alpha$ with $\alpha = w\alpha'$. Further, let $r = \delta^*_f(\epsilon,w)$. Then $\alpha \in U$ iff $r\alpha' \in U$.
\end{lemma}
\begin{proof}
We can write $\alpha = s\beta$ with $s \in S$ and $\beta \in E$. If $w$ is a prefix of $s$, the claim follows since $S$ is prefix-closed and thus $\delta_f^*(w) = w$. Otherwise, $\alpha'$ is a suffix of $\beta$ and hence $\alpha' \in E$ since $E$ is suffix-closed. Let $\alpha' = \sigma\alpha''$ and $r' = \delta_f(r,\sigma) = \delta_f^*(w\sigma)$. We show that $r\alpha' \in U$ iff $r'\alpha'' \in U$, which proves the claim by induction on the length of $w$. We have that $r\alpha' \in U$ iff $f(r,\alpha') = $ ``yes'' iff $f(r\sigma,\alpha'') = $ ``yes'' iff $f(r',\alpha'') = $ ``yes'' since $r'$ is the unique element in $S$ with $f_{r'} = f_{r\sigma}$.
\end{proof}

It remains to determine the accepting states of $\TSf$ in order to obtain a hypothesis WDBA. This is done by marking for each $\alpha \in S \cdot E$ the states in $\infs[\TSf](\alpha)$ with ``yes'' if $\alpha \in U$, and with ``no'' if $\alpha \notin U$. The marking fails if and only if there are two states $s,t$ in the same MSCC such that $s$ is marked ``yes'' and $t$ is marked ``no'' ($s$ and $t$ can be the same if there is a state that is marked in both ways). Otherwise, the marking succeeds and we ask an equivalence query for $\Hy = (\TSf,F)$ with
$F:=S \setminus \{q \in S \mid \text{ a state in MSCC$(q)$ is marked ``no''}\}$
the set of states that are not in the MSCC of a state that is marked ``no''.

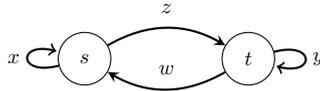
\begin{figure}[!h]
\vspace*{-2mm}
	\begin{center}
		\begin{tikzpicture} [node distance = 3cm, on grid,scale=0.5, every node/.style={scale=0.72}]
			\node (s) [state] {$s$};
			\node (t) [state, right of = s] {$t$};
			
			\path [-stealth, thick]
			(s) edge [loop left] node {$x$}()
			(s) edge [bend left] node [above=0.1cm] {$z$} (t)
			(t) edge [bend left] node [above=0.1cm] {$w$} (s)
			(t) edge [loop right] node {$y$}();
		\end{tikzpicture}
	\end{center}\vspace*{-5mm}
	\caption{The situation described in Lemma~\ref{lem:conflict} with $sx^\omega \in U \backslash D$ and $ty^\omega \not \in U \cup D$.} \label{fig:conflict}
\end{figure}

It remains to deal with the case that the marking fails. The situation described in the following lemma is illustrated in Figure~\ref{fig:conflict}. The proof is based on Lemma~\ref{lem:consistent}.

\begin{lemma} \label{lem:conflict}
If the marking fails, there are $s,t \in S$, and $z,w \in \Sigma^*$ with $\delta_f^*(sz)=t$ and $\delta_f^*(tw)=s$, as well as $x^\omega,y^\omega \in E$ with $s \in \infs[\TSf](sx^\omega)$ and $t \in \infs[\TSf](ty^\omega)$ such that $f(s,x^\omega)=\text{``yes''}$ and $f(t,y^\omega)=\text{``no''}$.
\end{lemma}
\begin{proof}
If the marking fails, there are $s,t \in S$ in the same MSCC such that $s$ is marked ``yes'' and $t$ is marked ``no''. Since $s,t$ are in the same MSCC, the are words $z,w$ connecting them. Furthermore, by definition of the marking, there are $\alpha,\beta \in S \cdot E$ with $\alpha \in U$ and $s \in \infs[\TSf](\alpha)$, and $\beta \notin U$ and $t \in \infs[\TSf](\beta)$. Since $\alpha,\beta$ are ultimately periodic, and $s$ and $t$ are visited infinitely often in the respective runs, we can decompose $\alpha$ and $\beta$ as $\alpha = ux^\omega$ and $\beta = vy^\omega$ with $\delta_f^*(u) = s$, $\delta_f^*(v) = t$, and $x^\omega,y^\omega \in E$. By Lemma~\ref{lem:consistent} we obtain that $sx^\omega \in U$ and $ty^\omega \notin U$, and thus $f(s,x^\omega)=\text{``yes''}$ and $f(t,y^\omega)=\text{``no''}$.
\end{proof}

From the situation described in Lemma~\ref{lem:conflict} we want to extend the observation table such that a new state is discovered. This is done by finding a \textit{distinguishing experiment}, which is of the form $rv\alpha \in \Sigma^\omega$ with $r \in S$ such that for $r':=\delta_f^*(rv)$ we have $r'\alpha \not \in U \iff rv\alpha \in U$. Adding $v\alpha$ and its suffixes to $E$ ensures by Lemma~\ref{lem:consistent} that $\delta_{f'}^*(rv) \not= r'$ in the next hypothesis with table $(S',E',f')$. Hence $\delta_{f'}^*(rv) \in S'\setminus S$, which means that the number of discovered states increases.
The algorithm in \cite{MALER1995316} can produce a distinguishing experiment with suffix $(zw)^\omega$, which might be a don't care word in our setting. We therefore have to identify a distinguishing experiment with a suffix that already appears in the table and thus cannot be a don't care word.

\begin{lemma} \label{lem:distinguish}
For $n:=Ind(\sim_{U,D})$ and the setting in Lemma~\ref{lem:conflict}, there is $i \le n$ such that one of the following is a distinguishing experiment: $s(x^nzy^nw)^ix^\omega$, $s(x^nzy^nw)^ix^nzy^\omega$, $t(y^nwx^nz)^iy^\omega$, or $t(y^nwx^nz)^iy^nwx^\omega$.
\end{lemma}
\begin{proof}
The $D$-minimal WDBA $\mathcal{A}=(\Sigma,Q,\delta,q_0,F)$ for $U$ has at most $n$ MSCCs, thus for all $w_1,w_2 \in \Sigma^*$, there is an integer $i \leq n$ such that $\delta^*(w_1w_2^j) \in \text{MSSC}(\delta^*(w_1w_2^i))$ for all $j \ge i$. Since every MSCC in $\mathcal{A}$ either consists entirely of final states or contains no final states, it follows that there must be an integer $i \leq n$ with
\[
\begin{array}{lcl}
s(x^nzy^nw)^ix^\omega \in U &\iff& s(x^nzy^nw)^ix^nzy^\omega \in U \text{ or }\\
 t(y^nwx^nz)^iy^\omega \in U &\iff&  t(y^nwx^nz)^iy^nwx^\omega \in U.
 \end{array}
\]
To see that, note that $\delta^*(s(x^nzy^nw)^ix^n)$ is in the MSCC of the infinity set of $s(x^nzy^nw)^ix^\omega$ by the above observation. So if, for example, $s(x^nzy^nw)^ix^\omega \in U$ and $s(x^nzy^nw)^ix^nzy^\omega \not\in U$, then the infinity set of $s(x^nzy^nw)^ix^nzy^\omega$ must be in another MSCC than the one of $s(x^nzy^nw)^ix^\omega$. This change of MSCC can happen at most $n$ times.

\medskip
We obtain the distinguishing experiment as follows. If $s(x^nzy^nw)^ix^\omega \not\in U$, then $rv\alpha$ with $r=s$, $v= (x^nzy^nw)^i$ and $\alpha=x^\omega$ satisfies the conditions with $r' = \delta_f^*(rv) = s$. If $s(x^nzy^nw)^ix^nzy^\omega \in U$, then choose $r=s$, $v = (x^nzy^nw)^ix^nz$ and $\alpha = y^\omega$ with $r' = \delta_f^*(rv) = t$. Similarly for the other cases.
\end{proof}

According to Lemma~\ref{lem:distinguish}, the following procedure finds a distinguishing experiment in time $\mathcal{O}(n^2)$. Since we do not know the value of $n$, we have to check the candidates from Lemma~\ref{lem:distinguish} for increasing values, called $k$ in the procedure. The procedure returns the word for which all suffixes need to be added to $E$:
\begin{enumerate}
\itemsep=0.95pt
\item Initialize $k=1$, then repeat the following steps:
\item Let $z':=x^kz$ and $w':=y^kw$.
	\item If $sz'y^\omega \in U$, terminate and return $x^kzy^\omega$
	\item If $tw'x^\omega \not \in U$, terminate and return $y^kwx^\omega$
	\item For $i \in \{1, \dots, k\}$ do
   \begin{enumerate}
		\item If $s(z'w')^ix^\omega \not \in U$, terminate and return $(z'w')^ix^\omega$
		\item If $t(w'z')^iw'x^\omega \not \in U$, terminate and return $(w'z')^iw'x^\omega$
		\item If $s(z'w')^iz'y^\omega \in U$, terminate and return $(z'w')^iz'y^\omega$
		\item If $t(w'z')^iy^\omega \in U$, terminate and return $(w'z')^iy^\omega$
	\end{enumerate}
	\item Set $k:=k+1$ and go back to step 2.
	
\end{enumerate}

We now have all the ingredients for the modified learning algorithm:
The algorithm initializes an observation table with $S=\{\epsilon\}$, $E=\emptyset$ and goes through the following steps:
\begin{enumerate}
	\item Add $t \in S\Sigma$ with $\forall s \in S: f_t \neq f_s$ to $S$ until the table is closed, and then construct $\TSf$.
	\item Mark the states in $\infs[\TSf](s\alpha)$ according to $f(s,\alpha)$ for every $s \in S$ and $\alpha \in E$. If the marking fails, find a distinguishing experiment, add its suffixes to $E$ and go back to step 1, else let $\Hy = (\TSf,F)$ with $F:=S \setminus \{q \in S \mid \text{ a state in MSCC$(q)$ is marked ``no''}\}$.
	\item If $\EQ({\mathcal{H}}$) returns a counterexample $\alpha$, add $\alpha$ and its suffixes to $E$ and go back to step 1, else return $\mathcal{H}$.
\end{enumerate} \vspace{0.25cm}
An example run of the algorithm can be found in the appendix.
\begin{theorem}\label{t2}
	Let $n:=Ind(\sim_{U,D})$. The algorithm learns a $D$-minimal WDBA $\mathcal{H}$ with $L(\mathcal{H}) \equiv_D U$ in time that is polynomial in $n$ and the length of a longest counterexample.
\end{theorem}
\begin{proof}
The algorithm starts with $E=\emptyset$ and thus the table is closed since all the functions $f_w$ have empty domain. By definition, the first hypothesis then accepts all words and the equivalence query returns a word outside $U$ and $D$ (if there is any, otherwise a $D$-minimal WDBA has been found).
By Lemma~\ref{lem:isomorphic} the transition system constructed by the algorithm is isomorphic to $\TUD$ once $|S|=n$. Then the marking will not fail anymore, so the algorithm will terminate with a correct hypothesis once it has for each MSCC at least one word with infinity set in that MSCC.

\medskip
It remains to show that the algorithm keeps discovering new states. If the table is not closed, a new element is added to $S$. So eventually, a transition system for a hypothesis will be constructed.
If the marking fails, then adding the suffixes of a distinguishing experiment ensures that a new state is discovered (this follows from Lemma~\ref{lem:consistent}, as explained earlier).

Otherwise, the suffixes of a counterexample are added to $E$. If the table remains closed, then the next marking either fails or in the resulting hypothesis at least one MSCC is marked differently than in the previous hypotheses for this transition system. This can happen at most $n$ times before the marking fails or the hypothesis is correct.
This means that the main loop terminates after at most $\mathcal{O}(n^2)$ rounds.
A distinguishing experiment can be found with at most $\mathcal{O}(n^2)$ membership queries (see Lemma~\ref{lem:distinguish} and the algorithm thereafter). The number of rows in the table is at most $n+n|\Sigma|$, and the number of columns is linear in the sum of the lengths of the counterexamples and the distinguishing experiments. By construction, the length of the distinguishing experiments is polynomial in $n$ and the length of the longest element in $E$ (see Lemmas~\ref{lem:conflict} and~\ref{lem:distinguish}).
\end{proof}

Note that the size of an automaton for $D$ does not play any role for the complexity of the learning algorithm. The learner does not know $D$, it only knows that $D$ has a trivial right-congruence. Based on this knowledge, the algorithm is designed in such a way that all membership queries are made for words outside $D$. The only way how $D$ affects the complexity of the learner is by the number of classes in $\sim_{U,D}$, which is at most the number of classes in $\sim_U$ (if two $\omega$-words are equivalent for $\sim_U$, they are clearly also equivalent for $\sim_{U,D}$). The counterexamples for equivalence queries are chosen by the oracle, so their lengths are not under control of the learner. However, since the hypothesis automata used in equivalence queries have at most $Ind(\sim_{U,D})$ many states, there always exists a counterexample of size polynomial in $Ind(\sim_{U,D})$.

\section{Conclusion}\label{sec:conclusion}

We have shown that the problem of priority optimization for DPA under a given don't care set can be solved efficiently, and that active learning of WDBA under a don't care set with trivial right-congruence is possible in polynomial time. Minimization of automata with IRC under such a don't care set $D$ is NP-hard, and $D$-minimal automata are not unique, which shows that DPA with IRC do not inherit all the good properties of WDBA. However, minimization without don't cares for DPAs that accept languages in $\mathbb{IP}$ is possible in polynomial time.

\medskip
Since DPAs are a model that is used in synthesis problems (see \cite{MeyerSL18} for a recent tool based on parity automata), it would be interesting to see if one can identify synthesis problems with don't cares in which the priority minimization can improve the results of such synthesis algorithms.



\vfill\eject

\appendix

\section{Appendix}

\subsection{Example Run of Learning Algorithm from Section~\ref{sec:learning}}
We illustrate the learning algorithm with an example execution. Let $\Sigma:=\{a,b,c\}$, $U:= ab^\omega + ba^\omega + (ab)^\omega$ and $D:=\Sigma^*b^\omega$. We show the closed observation tables and the hypotheses the algorithm produces. The entries in the observation table are denoted as $1$ for "yes" and $0$ for "no".

 \bigskip
 The first hypothesis for $S=\{\epsilon\}$ and $E=\emptyset$ is the WDBA accepting everything. Let $a^\omega$ be the counterexample, which is added to $E$:

\vspace{15mm}
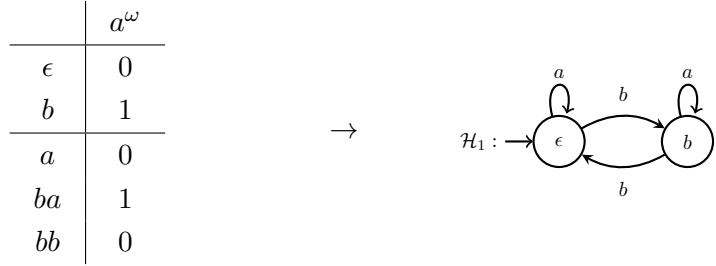
\begin{minipage}{0,4\textwidth}
	\centering
	\begin{tabular}{c|c}  
		
		& $a^\omega$ \\
		\hline
		$\epsilon$ & $0$ \\
		$b$ & $1$ \\
		\hline
		$a$ & $0$ \\
		$ba$ & $1$ \\
		$bb$ & $0$ \\
  \end{tabular}

\end{minipage}
\begin{minipage}{0,1\textwidth}
	$\rightarrow$
\end{minipage}
\begin{minipage}{0,5\textwidth}
\begin{tikzpicture}[thick,scale=0.7, every node/.style={scale=0.7}]
	\node (q0) [state, initial, initial text = {$\mathcal{H}_1:$}] {$\epsilon$};
	\node (q1) [state, right = of q0] {$b$};
	
	\path [-stealth, thick]
	(q0) edge [loop above] node {$a$}()
	(q0) edge [bend left] node [above=0.1cm] {$b$} (q1)
	(q1) edge [bend left] node [below=0.1cm] {$b$} (q0)
	(q1) edge [loop above] node {$a$}();
\end{tikzpicture}
\end{minipage}
\\

\vspace{15mm}

\par The marking fails and we have $s=b$, $t=\epsilon$, $x=a=y$, $z=b=w$. For $k=1$ and $w'=ab$, we find that $\epsilon aba^\omega \not \in U$. As such, we add $aba^\omega$ and its suffixes to the observation table: \\ \\ \vspace{0.5cm}

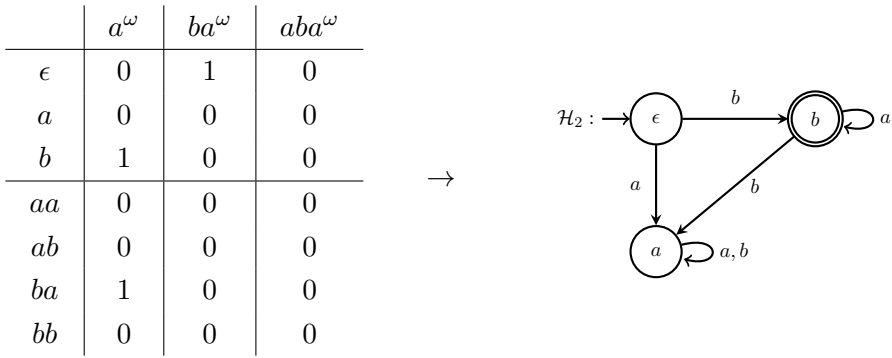
\begin{minipage}{0,4\textwidth}
	\begin{center}
	\begin{tabular}{c|c|c|c} 
		
		& $a^\omega$ & $ba^\omega$ & $aba^\omega$ \\
		\hline
		$\epsilon$ & $0$  & $1$ & $0$\\
		$a$ & $0$         & $0$ & $0$\\
		$b$ & $1$         & $0$ & $0$\\
		\hline
		$aa$ & $0$        & $0$ & $0$\\
		$ab$ & $0$        & $0$ & $0$\\
		$ba$ & $1$        & $0$ & $0$\\
		$bb$ & $0$        & $0$ & $0$\\
	\end{tabular}
	\end{center}	
\end{minipage}
\begin{minipage}{0,1\textwidth}
	$\rightarrow$
\end{minipage}
\begin{minipage}{0,5\textwidth}
	
	\begin{tikzpicture}[thick,scale=0.7, every node/.style={scale=0.7}]
		\node (q0) [state, initial, initial text = {$\mathcal{H}_2:$}] at (0,0) {$\epsilon$};
		\node (q1) [state] at (0,-2.5) {$a$};
		\node (q2) [state, accepting] at (3,0) {$b$};
		
		\path [-stealth, thick]
		(q0) edge node [left=0.1cm] {$a$} (q1)
		(q0) edge node [above=0.1cm] {$b$} (q2)
		(q1) edge [loop right] node {$a,b$}()
		(q2) edge node [right=0.1cm] {$b$} (q1)
		(q2) edge [loop right] node {$a$}();
		
	\end{tikzpicture}
\end{minipage}\\ \vspace{12mm}

	\vfil\eject

	The marking succeeds and we get $(ab)^\omega \in U$ as a counterexample. We extend the observation table accordingly and repeat the algorithm:\\ \vspace{6mm}

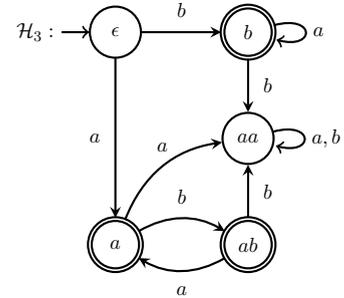
\begin{minipage}{0,5\textwidth}
	\begin{center}
		\begin{tabular}{c|c|c|c|c|c} 
			
			& $a^\omega$  & $ba^\omega$ & $aba^\omega$ & $(ab)^\omega$ & $(ba)^\omega$ \\
			\hline
			$\epsilon$ & $0$& $1$ & $0$ & $1$ & $0$\\
			$a$ & $0$       & $0$ & $0$ & $0$ & $1$\\
			$b$ & $1$       & $0$ & $0$ & $0$ & $0$\\
			$aa$ & $0$      & $0$ & $0$ & $0$ & $0$\\
			$ab$ & $0$      & $0$ & $0$ & $1$ & $0$\\
			\hline
			$ba$ & $1$      & $0$ & $0$ & $0$ & $0$\\
			$bb$ & $0$      & $0$ & $0$ & $0$ & $0$\\
			$aaa$ & $0$     & $0$ & $0$ & $0$ & $0$\\
			$aab$ & $0$     & $0$ & $0$ & $0$ & $0$\\
			$aba$ & $0$     & $0$ & $0$ & $0$ & $1$\\
			$abb$ & $0$     & $0$ & $0$ & $0$ & $0$\\
		\end{tabular}
	\end{center}
\end{minipage}
\begin{minipage}{0,1\textwidth}
	\begin{center}
	$\rightarrow$
	\end{center}
\end{minipage}
\begin{minipage}{0,4\textwidth}
	
	\begin{tikzpicture}[thick,scale=0.7, every node/.style={scale=0.7}] 
		\node (q0) [state, initial, initial text = {$\mathcal{H}_3:$}] at (0,0) {$\epsilon$};
		\node (q1) [state, accepting] at (0,-4) {$a$};
		\node (q2) [state, accepting] at (2.5,0) {$b$};
		\node (q3) [state] at (2.5,-2) {$aa$};
		\node (q4) [state, accepting] at (2.5,-4) {$ab$};
		
		\path [-stealth, thick]
		(q0) edge node [left=0.1cm] {$a$} (q1)
		(q0) edge node [above=0.1cm] {$b$} (q2)
		(q1) edge [bend left] node [above=0.1cm] {$a$} (q3)
		(q1) edge [bend left] node [above=0.1cm] {$b$} (q4)
		(q2) edge node [right=0.1cm] {$b$} (q3)
		(q2) edge [loop right] node {$a$}()
		(q3) edge [loop right] node {$a,b$}()
		(q4) edge [bend left] node [below=0.1cm] {$a$} (q1)
		(q4) edge node [right=0.1cm] {$b$} (q3);
		
	\end{tikzpicture}
\end{minipage}
\\ \vspace{8mm}
	
	The algorithm terminates and outputs $\mathcal{H}_3$ with $L_\omega(\mathcal{H}_3)$. $\mathcal{H}_3$ has informative $D$-congruence and is therefore $D$-minimal.

\end{document}